%% file: main.tex
\documentclass[11pt]{article}

\usepackage[utf8]{inputenc}

\usepackage[T1]{fontenc}    
\usepackage{hyperref}       
\usepackage{url}            
\usepackage{booktabs}       
\usepackage{amsfonts}       
\usepackage{nicefrac}       
\usepackage{microtype}      
\usepackage{xcolor} 
\usepackage{float}
\usepackage{fullpage}

\usepackage{todonotes}

\usepackage{graphicx, mathtools, amssymb, latexsym, amsmath, amsfonts, amsthm, bbm, tikz}

\usepackage{xcolor}
\usepackage{mathtools,amsfonts, mdframed,xspace,xfrac,multicol,bbm,bm}
\usepackage{makecell}
\usepackage{algorithm,algorithmic}
\usepackage{eucal}
\usepackage{amssymb,amsmath,amsopn,mathtools}
\usepackage{mathrsfs}
\usepackage{complexity}
\usepackage{enumitem}
\usepackage{subcaption}
\usepackage{dsfont}
\usepackage{wrapfig}
\usepackage{framed}
\usepackage{amsthm}

\newcommand{\slc}{{\sf Smooth-Label-Cover}\xspace}
\newcommand{\lc}{{\sf Label-Cover}\xspace}

\renewcommand{\E}{\mathbb{E}}

\renewcommand{\R}{\mathbb{R}}
\newcommand{\F}{\mathbb{F}}

\newcommand{\mc}[1]{\ensuremath{\mathcal{#1}}\xspace}
\newcommand{\mb}[1]{\ensuremath{\mathbf{#1}}\xspace}
\newcommand{\tn}[1]{\ensuremath{\textnormal{#1}}\xspace}
\newcommand{\ol}[1]{\ensuremath{\overline{#1}}\xspace}

\newcommand{\bc}{{\mb{c}}}

\newcommand{\llpparity}{{\sf LLP-PARITY}}
\newcommand{\llpor}{{\sf LLP-OR}}

\newcommand{\bx}{{\mathbf{x}}}
\newcommand{\by}{{\mathbf{y}}}

\newcommand{\bz}{{\mathbf{z}}}

\newtheorem{theorem}{Theorem}[section]
\newtheorem{lemma}[theorem]{Lemma}
\newtheorem{definition}[theorem]{Definition}

\newcommand{\vg}[1]{{}}
\newcommand{\rs}[1]{{}}

\title{Hardness of Learning Boolean Functions from Label Proportions\footnote{A conference version of this paper appeared in FSTTCS 2023.}}
\author{Venkatesan Guruswami\thanks{Department of EECS and Simons Institute for the Theory of Computing, University of California, Berkeley. {\tt venkatg@berkeley.edu}. Research supported in part by a Simons Investigator award and NSF grants CCF-2228287 and CCF-2211972.}
 \and Rishi Saket\thanks{Google Research India. {\tt rishisaket@google.com}}}
\date{}

\begin{document}

\maketitle

\begin{abstract}
    In recent years the framework of learning from label proportions (LLP) has been gaining importance in machine learning. In this setting, the training examples are aggregated into subsets or \emph{bags} and only the average label per bag is available for learning an example-level predictor. This generalizes traditional PAC learning which is the special case of unit-sized bags. The computational learning aspects of LLP were studied in recent works \cite{Saket21,Saket22} which showed algorithms and hardness for learning halfspaces in the LLP setting. In this work we focus on the intractability of LLP learning Boolean functions. Our first result shows that given a collection of bags of size at most $2$ which are consistent with an OR function, it is NP-hard to find a CNF of constantly many clauses which \emph{satisfies} any constant-fraction of the bags. This is in contrast with the work of \cite{Saket21} which gave a $(2/5)$-approximation for learning ORs using a halfspace.  Thus, our result provides a separation between constant clause CNFs 
    and halfspaces as hypotheses for LLP learning ORs.

\smallskip    
Next, we prove the hardness of satisfying more than $1/2 + o(1)$ fraction of such bags using a $t$-DNF (i.e. DNF where each term has $\leq t$ literals) for any constant $t$. In usual PAC learning such a hardness was known~\cite{KS08b} only for learning noisy ORs. 
We also study the learnability of parities and show that it is NP-hard to satisfy more than $(q/2^{q-1} + o(1))$-fraction of $q$-sized bags which are consistent with a parity using a parity, while a random parity based algorithm achieves a $(1/2^{q-2})$-approximation.
\end{abstract}

\input{introduction}

\input{preliminaries}

\input{llp_OR_ell_clause_CNF}

\input{llp_OR_ell_DNF}

\input{llp_parity}

\bibliographystyle{abbrv}

\bibliography{refsLLPbool}

\end{document}

%% file: introduction.tex
\section{Introduction}
In common machine learning applications, one is required to train a classifier using some training set  of (vectors, label)-pairs to predict the label of vectors sampled from the same (or a similar) distribution as the training set. A typical approach is to optimize the classifier to predict correctly on the training set to ensure that the classifier has good predictive performance over the target distribution. This optimization view is captured by the \emph{probably approximately correct} (PAC) learning framework~\cite{Valiant}.

In setting of \emph{learning from label proportions} (LLP), the training set consists of subsets or \emph{bags} of vectors along with the sum or average of the labels of vectors in each bag. The goal is to train a model to predict the labels for vectors. As before, one would want the model to firstly predict as correctly as possible on the training bags. One measure of such performance is the fraction of \emph{satisfied} bags i.e., those on which the predicted average label matches the given average label i.e., the label proportion. Note that traditional PAC learning is the special case of LLP with only unit-sized bags. 

LLP is motivated by applications in which only the aggregated labels for bags of vectors are available. This may be to preserve the privacy~\cite{R10,WIBB,o2022challenges} of labels, due to lack of instrumentation to obtain labels~~\cite{DNRS} or high labeling costs~\cite{CHR}. 
Other examples of LLP applications have been in medical image classification~\cite{hernandez2018,Bortsova18, Orting16} where small bag sizes -- in the range of 10 to 50 -- are typically more relevant (see Sec 1.2 of \cite{BK21}).

The work of \cite{Saket21} studied LLP from the computational learning perspective on bags of size $\leq 2$. The \emph{LLP learning} goal is the following: given a collection of bags consistent with some function from a target concept class, compute a hypothesis satisfying the most number of bags. With this objective, \cite{Saket21} showed a $(1/2 + o(1))$-factor hardness for LLP learning a halfspace using any function of constantly many halfspaces on bags of size at most $2$. From the algorithmic side on such bags \cite{Saket21} gave a $(2/5)$-factor approximation for LLP learning a halfspace using halfspace, based on rounding a semi-definite programming (SDP) relaxation. Subsequently, \cite{Saket22} proved a strengthened $(4/9 + o(1))$-factor hardness for LLP learning a halfspace using any function of constantly many halfspaces on bags of size at most $2$, a corresponding $(1/q + o(1))$-factor hardness for bags of size at most any constant $q \in \mathbb{Z}^+$, and extended the algorithmic result of \cite{Saket21} showing a $(1/12)$-approximation on bags of size at most $3$.

Since halfspaces capture {\sf OR} formulas,  the algorithmic results of \cite{Saket21,Saket22} apply to learning {\sf OR} formulas using halfspaces. Moreover, the $(1/2 + o(1))$-factor hardness of \cite{Saket21} on bags of size $\leq 2$ also holds for LLP learning an {\sf OR} using any function of constantly many halfspaces. Typically however, one would like to learn an {\sf OR} using an {\sf OR} or similar Boolean functions such as $\ell$-clause {\sf CNF} formulas ({\sf OR} is $1$-clause {\sf CNF}), rather than halfspaces. This raises the following question

\smallskip

\noindent
\emph{Can we achieve constant-factor algorithmic approximations for LLP learning {\sf OR} using {\sf OR} or constant-clause {\sf CNF}?}

\smallskip
\noindent
In our first result, we answer the above question in the negative. 
\begin{theorem} \label{thm:main-1}
    For any constants $\delta > 0, \ell \in \mathbb{Z}^+$, given a collection of bags which are of size at most $2$, and whose label proportions are consistent with some {\sf OR}, it is NP-hard to compute an $\ell$-clause {\sf CNF} that satisfies $\delta$-fraction of the bags. 
\end{theorem}
\noindent
The above theorem is proved in Sec. \ref{sec:LLPOR_CNF}. We find the result interesting since it (along with the algorithmic results of \cite{Saket21,Saket22}) proves a separation between constant clause {\sf CNF}s -- in particular {\sf OR}s -- and halfspaces as hypotheses for learning ORs. 

We also study the LLP learnability of {\sf OR} using as hypothesis $\ell$-{\sf DNF} formulas i.e., {\sf DNF} where each term is a conjunction of at most $\ell$ literals. While {\sf OR} is $1$-{\sf DNF}, for $\ell \geq 2$, $\ell$-{\sf DNF}s are not contained in halfspaces, therefore have the possibility of yielding better approximations. However, our second result below (proved in Sec. \ref{sec:LLPOR_ellDNF}) essentially rules out this possibility.
\begin{theorem}\label{thm:main-2}
    For any constants $\delta > 0, \ell \in \mathbb{Z}^+$, given a collection of bags which are of size at most $2$, and whose label proportions are consistent with some {\sf OR}, it is NP-hard to compute an $\ell$-{\sf DNF} that satisfies $(1/2 + \delta)$-fraction of the bags.
\end{theorem}
Note that while the hardness factor achieved above is weaker than that of Theorem \ref{thm:main-1},
no inapproximability is known for the real analogue of Theorem \ref{thm:main-2} i.e., LLP learning halfspaces using \emph{polynomial thresholds}.

While the works of \cite{Saket21,Saket22} studied the LLP learnability of halfspaces, the corresponding problem over finite fields has not been studied. In particular, the $\F_2$-version of this problem is equivalent to the LLP learnability of parities using parities over the Boolean domain. Parities are a fundamental class of Boolean functions which makes this problem of significant interest as well. Our next result however, shows that this is hard to approximate, with the inapproximability growing exponentially as the bag size increases.
\begin{theorem}\label{thm:main-3}
    For any constants $\delta > 0, q \in \mathbb{Z}^+$ ($q \geq 2$), given a collection of bags which are of size at most $q$, and whose label proportions are consistent with some parity, it is NP-hard to compute a parity that satisfies $(q/2^{q-1} + \delta)$-fraction of the bags.
\end{theorem}
The above is proved in Section~\ref{sec:paritieshardness}. The hardness factor is asymptotically close (for large $q$) to the following $(1/2^{q-2})$-approximation for this problem described in Sec. \ref{sec:algoparities}.
\begin{theorem}\label{thm:main-algo}
    For any $q \in \mathbb{Z}^+$, $q \geq 2$, given a collection of bags which are of size at most $q$, and whose label proportions are consistent with some parity, there is a randomized polynomial time algorithm that satisfies $(1/2^{q-2})$-fraction of the bags in expectation. 
\end{theorem}
Note that, when $q=1$ i.e., all bags are of size $1$, one can deterministically satisfy all bags using Gaussian elimination.

\subsection{Previous Related Work}
The formalization of the LLP framework was first done in the work of \cite{YCKJC14} who proved generalization error bounds for classifiers for any distribution over $($bag, label-proportion$)$-pairs, though their bag-level objective was a relaxed notion -- useful for studying LLP with large bag sizes -- of the strict bag satisfaction used in \cite{Saket21,Saket22} and our work. Related recent works  \cite{busafekete2023easy,chen2023learning} have shown bag-to-instance classification generalization error bounds. The study of LLP learnability of specific function classes has nevertheless been fairly sparse, apart from the works of \cite{Saket21,Saket22} whose contributions have been described earlier in this section.

The learnability of small Boolean formulas has been extensively studied in traditional PAC learning. It is well known that an {\sf OR} can be efficiently learnt by an {\sf OR} up to arbitrary accuracy. On the other hand, \cite{KS08b} proved a $(1/2 + o(1))$-factor hardness for learning a $2$-clause {\sf CNF} using constant clause {\sf CNF}, and the same hardness factor for learning a noisy {\sf OR} using $\ell$-{\sf DNF} for any constant $\ell$. The work of \cite{FGRW12} proved the same hardness for learning noisy {\sf OR} with a halfspace as hypothesis.  These results were further generalized by \cite{GS19} who proved the same hardness factors for learning $2$-clause {\sf CNF} and noisy {\sf OR} using any function of constantly many halfspaces as hypothesis. 
Similar to {\sf OR}, parities can also be efficiently learnt by parities using Gaussian elimination over $\F_2$. On the other hand, the $(1/2 + o(1))$-factor hardness for noisy {\sc Max}-$3$-{\sc Lin} by \cite{Hastad} implies the same hardness factor for learning a noisy parity using a parity. 
Note that all the $(1/2 + o(1))$-factor hardness results are tight since one of the constant $0$ or $1$ functions trivially obtain $(1/2)$-approximation for learning Boolean valued functions. However, this trivial threshold does not hold in the LLP setting since the constant functions are not guaranteed to satisfy even one bag.

The above hardness results carry over to the LLP setting for the special case when all bags are unit-sized. However, we prove hardness of approximating the problems of LLP learning OR and parity \emph{without} any noise, which are tractable in the usual PAC case, thereby showing a qualitative difference between the LLP and PAC settings.

\subsection{Overview of Our Techniques}
{\bf Proof of Theorem \ref{thm:main-1}.} Our reduction is from bipartite Label-Cover~\cite{ABSS97} with $N$ and $M$ as the sizes of the smaller and larger label sets respectively,  and is similar to that of \cite{KS08b} for the hardness of learning noisy {\sf OR}. The high-level approach is to have one coordinate for each vertex-label pair on the larger (right) side of the Label-Cover instance, i.e. the variables are $x_{v,i}$ for $v \in V$ and $i \in [M]$. Fix a random sample of $2t$ vertices $\{\hat{v}_1, \dots, \hat{v}_t, \tilde{v}_1, \dots, \tilde{v}_t\}$ from a neighborhood of a left vertex $u$. For simplicity assume that the projection constraints between $u$ and each of the $2t$ vertices are the same i.e, for each label $j \in [N]$ for $u$ there is a subset $S_j\subseteq[M]$ such that assigning any of the $2t$ vertices with a label from $S_j$ satisfies that edge.

A  $2$-sized bag with label proportion $1/2$ is sampled by letting $J \subseteq [N]$ be a random subset, and for the first point $\bx$ setting only the coordinates $\{x_{\hat{v}_r, i}\,\mid\, \pi_{\hat{v}_ru}(i) \in J, r\in [t]\}$ to be $1$, and for the second point $\bz$ only the coordinates  $\{z_{\tilde{v}_r, i}\,\mid\, \pi_{\tilde{v}_ru}(i) \in \ol{J}, r\in [t]\}$ to be $1$. It is easy to see in the YES case that an {\sf OR} of exactly the coordinates $(v,\rho(v))$ -- where $\rho$ is the satisfying labeling -- for each right vertex $v$, satisfies all such bags. For the NO case, we illustrate the analysis of an {\sf OR} formula $\mc{C}$ which satisfies some constant fraction of the bags. From the $o(1)$-Hamming weight of the points, one can assume that $\mc{C}$ has no negated literals. If $\mc{C}$ has no coordinates of the $2t$ vertices (empty case)  then the bag is anyway not satisfied as $\mc{C}$ evaluates to $0$ on both points. 

On the other hand, if a sufficiently large number of these vertices have a corresponding variable in $\mc{C}$ (dense case), then elementary probabilistic arguments yield at least two vertices among the $2t$ which have their pre-decided distinguished variables in $\mc{C}$
with the same projection, leading to a good randomized labeling to the Label-Cover. A key idea for ensuring that this analysis goes through is to sample $t$ u.a.r. from $\{1,\dots, 2^T\}$ for some large $T$ so that for each $u$ and most values of $t$, nearly all samples of the $2t$ vertices are either the empty case or the dense case.

\medskip
\noindent
{\bf Proof of Theorem \ref{thm:main-2}.} The hardness reduction is similar to the above, except that we need to ensure that a significant fraction of vertices have at least one term in which all the positive literals correspond to its coordinates, while none of the negated literals do. Thereafter a similar analysis as the previous case goes through. However, for ensuring this property we introduce an additional distribution over bags of size $1$ and label proportion $1$, essentially saying that for each vertex the point with all its coordinates to $1$ and the rest to $0$ should be $1$-labeled. Due to this we obtain a $(1/2 + o(1))$-factor hardness in this case.

\medskip
\noindent
{\bf Proof of Theorem \ref{thm:main-3}.} While the reductions above create points in a bag whose active coordinates span the edges of the label-cover, in the parity case we can add homogeneous $\F_2$-linear \emph{folding} constraints which ensure consistency of labels across edges via a reduction from the non-bipartite \emph{Smooth} Label-Cover~\cite{GRSW}. It is sufficient to then describe a \emph{dictatorship test} (see Chap. 7 of \cite{odonnell}) on the $M$ coordinates of a single vertex. Our dictatorship test is a distribution over $q$-sized bags, i.e., $q$ points in the hypercube $\F_2^M$, sampled as follows: independently for each $i \in M$, set exactly one of the $q$ points to $1$ in the $i$'th coordinate and the rest to $0$. All these bags have target label proportion $1/q$ which is satisfied by any dictator, i.e., parity given by a single coordinate. On the other hand, any parity over a much larger number $K$ of coordinates will induce a near-uniform distribution over the  $q$-sized vector of labeling to the points of a random bag. 
In fact, this is  close to the uniform distribution over the points of $\F_2^q$ with even or odd (depending only on $q$ and $K$) number of non-zero coordinates. It is then easy to see that such a distribution will satisfy the $1/q$ label proportion of the bags with probability $\approx q/2^{q-1}$. 

The algorithm for this problem first does Gaussian elimination for all the linear constraints given by bags of label proportion $\{0,1\}$ and obeying the bag-level parity constraints for the remaining bags. It then chooses a random parity from the remaining coordinates. We show that this yields an $1/2^{q-2}$ approximation.

%% file: preliminaries.tex
\section{Preliminaries}
\subsection{Problem Definitions}

Consider the space $\{0,1\}^d$ for some $d \in \mathbb{Z}^+$ and some function $f : \{0,1\}^d \to \{0,1\}$. For $B \subseteq \{0,1\}^d$, define $\sigma(B, f) := \left|\left\{ \bx \in B\,\mid\, f(\bx) = 1 \right\}\right|/|B|$ to be the corresponding label proportion.

An instance $\mc{I}$ of \llpor$[q]$ is given by a collection $\mc{B} := \{(B_j, \sigma_j)\}_{j=1}^m$ of bags and their label proportions where each bag is of size at most $q$. 
The goal is to find an {\sf OR} function $h(\bx)$ which satisfies the most bags of $\mc{B}$ i.e., maximize $\left|\left\{j\in [m]\,\mid\, \sigma_j = \sigma(B_j, h)\right\}\right|$. 

An instance $\mc{I}$ of \llpparity$[q]$ is similar to the above except that the 
goal accordingly is to compute a parity maximizing the number of satisfied bags. Since the XOR is simply an addition over $\F_2$ we shall think of the Boolean values as elements of $\F_2$ in this case.

\subsection{Label Cover}

\begin{definition}\label{def-LC}
An instance of $\mc{L}$ of  \lc is given by $(G(U, V, E \subseteq V\times U), M, N, \{\pi_{vu} : [M] \to [N]\}_{e = (v,u) \in E})$ where $G(U, V, E)$ is a bi-regular bipartite graph. A labeling $\rho$ assigning labels from $[N]$ to $U$ and $[M]$ to $V$ satisfies an edge $(u, v)$ iff $\pi_{vu}(\rho(v)) = \rho(u)$. The goal is to find a labeling satisfying the most number of edges.
\end{definition}

The following well-known inapproximability of \lc follows from the the PCP Theorem~\cite{AS,ALMSS} along with the Parallel Repetition Theorem~\cite{Raz}.
\begin{theorem}\label{thm:lc-hardness}
    For any constant $\xi > 0$ there exist $M$ and $N$ such that it is NP-hard, given an \lc instance $\mc{L}(U, V, E, M, N, \{\pi_{vu}\}_{e = (v,u) \in E})$ to distinguish between: (i) \tn{YES case:} there is labeling that satisfies all edges in $E$, or (ii) \tn{NO case:} Any labeling satisfies at most $\xi$ fraction of the edges of $E$.
\end{theorem}

\subsection{Smooth Label Cover}

\noindent
Unlike the standard bipartite version in Definition \ref{def-LC} we also use a non-bipartite version with useful structural properties defined below.
\begin{definition}\label{def-SLC}
An instance of \slc
$\mc{L}(G(V,E),N,M,\{\pi_{ev}\,\mid\,e\in E, v\in e\})$ consists of a regular
connected (undirected)
graph $G(V,E)$
with vertex set $V$ and edge set $E$.  Every edge
$e = (v_1,v_2)$ is associated with  projection functions
$\{\pi_{ev_i}\}_{i=1}^{2}$ where $\pi_{ev_i}: [M] \to [N]$.
A vertex labeling is a mapping defined on $\rho : V\to [M]$. A
labeling $\rho$ satisfies edge $e = (v_1, v_2)$ if
$\pi_{ev_1}(\rho(v_1)) = \pi_{ev_2}(\rho(v_2))$. The goal is to
find a labeling which satisfies the maximum number of edges.
\end{definition}
The following theorem is proved in Appendix A of \cite{GRSW}. 
\begin{theorem}\label{thm:slc-hardness}
There exists a constant $c_0 > 0$ such that for any constant
integer parameters $Q,R\geq 1$,
it is {\rm NP}-hard to distinguish between the following
two cases for a \textnormal{Smooth Label Cover} instance
$\mc{L}(G(V,E),N,M,\{\pi_{ev}\,\mid\,
e\in E, v\in e\})$ with $M = 7^{(Q+1)R}$ and $N= 2^{R}7^{QR}$:
\begin{itemize}
	\item \textnormal{(YES Case)} There is a labeling that satisfies every
edge.
\item \textnormal{(NO Case)}  Every labeling satisfies less than a fraction $2^{-c_0 R}$ of the edges.
\end{itemize}
In addition, the instance
$\mc{L}$ satisfies the following properties:
\begin{itemize}
	\item \textnormal{(Smoothness)} For any vertex $w \in V$, $ \forall
i,j\in [M],\ i\neq j, \ \ \Pr_{e\sim w} \left[\pi_{ew}(i) = \pi_{ew}(j)\right] \leq
1/Q,$ where the probability is over a randomly chosen edge incident
on $w$.
\item \textnormal{(Weak Expansion)} For any $\delta > 0$,  let $V'\subseteq V$ and
$|V'| = \delta\cdot |V|$, then the number of edges among the vertices
in $|V'|$ is at least $\delta^2|E|$.
\end{itemize}
\end{theorem}

\subsection{Tail Bounds}
We use the Chernoff bound stated as follows.
\begin{lemma}[Chernoff Bound] 
\label{lem:chernoff}
Suppose $X_1, \dots, X_n$ and independent $\{0,1\}$-valued random variables with $S = \sum_{i=1}^n X_i$ and $\mu = \E[S]$. Then, $\Pr[S \leq (1-\delta)\mu] \leq \tn{exp}(-\delta^2\mu/2)$ for any $\delta > 0$.
\end{lemma}

%% file: llp_OR_ell_clause_CNF.tex
\section{Hardness of LLP Learning {\sf OR} using $\ell$-clause {\sf CNF}}\label{sec:LLPOR_CNF}

We prove the following  hardness reduction which along with Theorem \ref{thm:lc-hardness} implies Theorem \ref{thm:main-1}.

\begin{theorem}
    For any constants $\delta > 0$ and $\ell \in \mathbb{Z}^+$, there is a polynomial time reduction from an instance $\mc{L}$ of \lc  to an \llpor$[2]$ instance $\mc{B}$ s.t. \\
    \tn{YES Case:} If $\mc{L}$ is YES instance then there is an {\sf OR} consistent with all the bags of $\mc{B}$. \\
    \tn{NO Case:} If $\mc{L}$ is a NO instance then there is no  $\ell$-clause {\sf CNF} formula satisfying at least $\delta$-fraction of the bags. 
\end{theorem}

We begin with the following useful technique.
Let $T \geq 10$ be a large integer and consider the set $\mc{T} = \{2, 4, \dots, 2^T\}$. For any $s \in \R^+$ define the subsets $L(s), R(s) \subseteq \mc{T}$ as
\begin{equation}
    L(s) := \{t \in \mc{T}\,\mid \, t \leq s/\sqrt{T}\} \quad \tn{ and } \quad R(s) := \{t \in \mc{T}\,\mid \, t \geq s\cdot\sqrt{T}\}. \label{eqn:LsRs}
\end{equation}
We have the following simple lemma:
\begin{lemma} \label{lem:LsRs}
    For any $s \in \R^+$, $L(s) \cap R(s) = \emptyset$ and $\left|L(s)\right| + \left|R(s)\right| \geq \left|\mc{T}\right| - 2\log T$.
\end{lemma}
\begin{proof}
    Since $T > 1$, $L(s) \cap R(s) = \emptyset$ by definition. Let $t' \in \mc{T}$ be the smallest element which is larger than $s/\sqrt{T}$ and $t'' \in \mc{T}$ be the largest element smaller than $s\sqrt{T}$. By definition, we have that $\mc{T}\setminus\left(L(s)\cup R(s)\right) = [t', t'']\cap\mc{T}$. Note also that $t''/t' \leq T$, and thus $\left|[t', t'']\cap\mc{T}\right| \leq \log T + 1 \leq 2\log T$ since $T \geq 10$, which completes the proof.
\end{proof}

\subsection{Hardness Reduction}
The hardness reduction is from a \lc instance $\mc{L}(U, V, E \subseteq V\times U, M, N, \{\pi_{vu} : [M] \to [N]\}_{e = (v,u) \in E})$. We shall use $\mc{T}$ as defined above for some large enough choice of $T$ depending on $\delta$ and $\ell$. Note that $T$ is a constant compared to $|V|$ which is an increasing value. The underlying space of the vectors is $\mc{X} = \{0,1\}^{V\times [M]}$ i.e., a vector $\bx \in \mc{X}$ is given by $\bx = \left(x_{v,i}\right)_{v \in V, i \in [M]}$. 
The reduction yields a distribution $\mc{D}_{\mc{B}}$ over $2$-sized bags and all bags $B$ in its support have the label proportion $\sigma = 1/2$. A random bag from $\mc{D}_{\mc{B}}$ is given by the following steps:
\begin{enumerate}
    \item Sample $t$ uniformly at random from $\mc{T}$.
    \item U.a.r. sample a vertex $u \in U$.
    \item Independently and u.a.r. sample vertices $V^x_u = \{\hat{v}_1, \dots, \hat{v}_t\}$ and $V^z_u = \{\tilde{v}_1, \dots, \tilde{v}_t\}$ from the neighborhood $N(u)$ of $u$ in $V$.
    \item Randomly sample  $J \subseteq [N]$, and let $\ol{J} = [N]\setminus J$. 
    \item Define a point $\bx \in \mc{X}$ as follows. For each $i \in [M]$ and $v \in V$ set:
    \begin{equation}
        x_{v,i} = \begin{cases}
        1 & \tn{ if } v \in V^x_u \tn{ and } \pi_{vu}(i) \in J \\
        0 & \tn{ otherwise.}
        \end{cases}
    \end{equation}
    \item Define another point $\bz \in \mc{X}$ as follows. For each $i \in [M]$ and $v \in V$ set:
    \begin{equation}
        z_{v,i} = \begin{cases}
        1 & \tn{ if } v \in V^z_u \tn{ and } \pi_{vu}(i) \in \ol{J} \\
        0 & \tn{ otherwise.}
        \end{cases}
    \end{equation}
    \item Output $\left(B = \{\bx, \bz\}, \sigma_B = 1/2\right)$.
\end{enumerate}
In particular, observe that the points $\bx$ and $\bz$ are zero outside of the coordinates corresponding to the vertices in $V^x_u\cup V^z_u$. 

\medskip
\noindent
{\bf YES Case.} Consider a labeling $\rho$ to the vertices of $\mc{L}$ that satisfy all the edges. Define the {\sf OR}, $h^*(\bx) = \bigvee_{v \in V}x_{v,\rho(v)}$. Let $u$ be the choice in Step 2 above, and assume that $\rho(u) \in J$ as chosen in Step 4. Then, we know that for all $v \in V^x_u\cup V^z_u$ (as chosen in Step 3), $\pi_{vu}(\rho(v)) = \rho(u) \in J$. By construction of $\bx$ and $\bz$ therefore, $h^*(\bx) = 1$ and $h^*(\bz) = 0$, and thus $B$ is satisfied by $h^*$. Similarly, when $\rho(u) \in \ol{J}$, we obtain that $h^*(\bx) = 0$ and $h^*(\bz) = 1$.

\subsection{NO Case}
Assume for a contradiction an $\ell$-clause {\sf CNF} formula $h'$ s.t. $\Pr_{B \leftarrow \mc{D}_{\mc{B}}}[h' \tn{ satisfies } B] \geq \delta$. From the bi-regularity of $\mc{L}$,
if any clause of $h'$ contains a negated literal then with probability at least $1 - 2t/|V|$ the literal's coordinate is not from those of the vertices in $V^x_u\cup V^z_u$ and the clause evaluates to $1$ on both points of a random bag $B \leftarrow \mc{D}_{\mc{B}}$. Removing all such clauses  we obtain $r$-clause {\sf CNF}  $h = C_1\wedge\dots \wedge C_r$ ($r \leq \ell$) that satisfies at least $\delta - 2t\ell/|V| \geq \delta/2$ fraction of the bags (since we can take $|V| \gg 2^{T+2}\ell/\delta$ as $|V|$ is super-constant). We have the following lemma.
\begin{lemma} \label{lem:clause0or1}
    For any constant $\zeta > 0$, there is a choice of $T = T(\zeta)$ s.t. for any $C_i$ ($i \in [r]$),
        $\Pr_{B = (\bx,\bz) \leftarrow \mc{D}_{\mc{B}}}\left[C_i(\bx) \neq C_i(\bz)\right] \leq \zeta$.
    
\end{lemma}
Using the above lemma, the NO case proof can be completed by taking $\zeta = \delta/(6\ell)$. By a union bound, the probability that any one of $C_1, \dots, C_r$ evaluates differently on $\bx$ and $\by$ is at most $\delta/6$. This also upper bounds the probability of satisfying the bags of $\mc{D}_{\mc{B}}$, contradicting  our assumption.

\noindent
\begin{proof}[Proof of Lemma \ref{lem:clause0or1}]
Fix any clause $C \in \{C_1, \dots, C_r\}$, and from our construction above $C$ has no negated literals. Call a vertex which has at least one variable from $C$ as \emph{non-empty}, otherwise call it \emph{empty}. For each non-empty vertex $v$ arbitrarily choose $i_v$ such that the $(v,i_v)$-th variable is in $C$. For each $u \in U$ let $\mu(u)$ denote the fraction of its neighbors which are non-empty. In our analysis below we shall be collecting the \emph{error} probabilities using
\begin{eqnarray}
    \Pr\left[A\right] \leq \Pr\left[A\cap B\right] + \Pr\left[\bar{B}\right], \tn{ and } \Pr\left[A\cap B\right] \leq \min\left\{\Pr\left[A\right],\Pr\left[A\mid B\right]\right\}  \label{eqn:collect-error}
\end{eqnarray}
for any two events $A$ and $B$, where $\bar{B}$ denotes the complement of $B$.

We will first bound $\gamma$ which we define to be the probability over the choice of $t$, $u$, $V^x_u$ and $ V^z_u$ that there is a pair of non-empty vertices $v, v' \in V^x_u\cup  V^z_u$ s.t. $\pi_{vu}(i_v) = \pi_{v'u}(i_{v'})$. We call this event $\Psi$. In this case we can construct a randomized partial labeling $\rho$ for the vertices of $\mc{L}$ as follows: for each $v \in V$ which is non-empty, assign it the label $i_v$ defined above. For each $u$, select a random neighbor $v_u$ and assign $u$ the label $\pi_{v_u u}(i_{v_u})$ if $v_u$ is non-empty. Since $t \leq 2^T$ we obtain that this randomized labeling satisfies in expectation at least $\max_t \gamma/(2t)^2 = \gamma/\left(2^{2(T+1)}\right)$ fraction of the edges. Choosing the soundness of $\mc{L}$ to be small enough one can ensure that $\gamma \leq \zeta/100$. 

We consider two the cases for $t$ and $u$ below, and in each of them (setting $B = \Psi$ in \eqref{eqn:collect-error}) we can assume that $\Psi$ does not occur, while incurring at most $\zeta/100$ probability error.

\medskip
\noindent
{\bf Case I: $t \in L(1/\mu(u))$.} In this case the probability that there is non-empty vertex in $V^x_u\cup V^z_u$ is at most $2t\mu(u) \leq 2\left(1/\sqrt{T}\right)\left(1/\mu(u)\right)\mu(u)\leq 2/\sqrt{T}$. Therefore, except with probability at most $2/\sqrt{T}$, $C$ evaluates to zero on $\bx$ and $\bz$.

\medskip
\noindent
{\bf Case II: $t \in R(1/\mu(u))$.} 
We shall first show that w.h.p. $V^x_u$ contains a significant number of non-empty vertices, and given that happens, w.h.p. $C$ evaluates to $1$ on $\bx$. 
The expected number of non-empty vertices in $V^x_u$ is $t \mu(u) \geq \left(\sqrt{T}/\mu(u)\right)\mu(u) \geq \sqrt{T}$. Therefore, by the Chernoff bound (see Lemma~\ref{lem:chernoff}), except with probability $\tn{exp}\left(-\sqrt{T}/8\right)$, the number of non-empty vertices is at least $\sqrt{T}/2$. 
From the assumption that $\Psi$ does not occur, for all pairs of non-empty vertices $v, v' \in V^x_u$, $\pi_{vu}(i_v) \neq \pi_{v'u}(i_{v'})$. Thus, each $\{x_{v,i_v}\,\mid\, v \in V^x_u, v \tn{ is non-empty }\}$ is independently set to $1$ with probability $1/2$. In particular, except with probability $(1/2)^{\sqrt{T}/2}$, at least one of $\{x_{v,i_v}\,\mid\, v \in V^x_u, v \tn{ is non-empty }\}$ is set to $1$ and thus $C$ evaluates to $1$ on $\bx$. 
The same argument as above also works for $ V^z_u$ and $\bz$. 

\medskip
From the analysis of the above two cases, using Lemma \ref{lem:LsRs} and repeated applications of \eqref{eqn:collect-error} to add up the error probabilities above (for $V^z_u$ and $\bz$ as well in Case II) we obtain that:
\begin{equation}
    \Pr\left[C(\bx) \neq C(\bz)\right] \leq 2/\sqrt{T} + \zeta/100 + 2\left(\frac{\log T}{T} +  \tn{exp}\left(-\sqrt{T}/8\right) + 2^{-\sqrt{T}/2}\right)
\end{equation}
Choosing $T$ to be $100/\zeta^2$ we can ensure that the above probability is at most $\zeta$.
\end{proof}

%% file: llp_OR_ell_DNF.tex
\section{Hardness of LLP Learning {\sf OR} using $\ell$-{\sf DNF}} \label{sec:LLPOR_ellDNF}
This section proves the following hardness reduction which implies Theorem \ref{thm:main-2}.

\begin{theorem}
    For any constants $\delta > 0$ and $\ell \in \mathbb{Z}^+$, there is a polynomial time reduction from an instance $\mc{L}$ of \lc  to an \llpor$[2]$ instance $\mc{B}$ s.t. \\
    \tn{YES Case:} If $\mc{L}$ is YES instance then there is an {\sf OR} consistent with all the bags of $\mc{B}$. \\
    \tn{NO Case:} If $\mc{L}$ is a NO instance then there is no  $\ell$-{\sf DNF} formula satisfying at least $(1/2 + \delta)$-fraction of the bags. 
\end{theorem}
\subsection{Hardness Reduction}
The setup is the same as in the previous section. 
The reduction outputs a distribution $\mc{D}_{\mc{B}}$ over two types of bags with equal probability: (i) the first type has size $1$ with label proportion $1$, (ii) the second type has size $2$ and label proportion $1/2$. With $T$ being a large enough constant to be chosen later, the following steps define a random bag of $\mc{D}_{\mc{B}}$.
\begin{enumerate}
    \item U.a.r. sample a vertex $u \in U$.
    \item With probability $1/2$ do the following:
    \begin{enumerate}
        \item U.a.r. sample vertex $v \in N(u)$ and create the point $\ol{\bx}^{(v)}$ as follows: set all coordinates $\{\ol{x}^{(v)}_{v,i}\}_{i=1}^M$ to $1$, and set all the other coordinates to $0$.
        \item Output $\left(B = \{\ol{\bx}^{(v)}\}, \sigma_B = 1\right)$.
    \end{enumerate}
    \item With the remaining probability $1/2$ do the following:
    \vg{The $\hat{V}$ and $\tilde{V}$ are a bit hard to differentiate.}\rs{Changed the notation to $V^x$ and $V^z$}
    \begin{enumerate}
        \item Independently and u.a.r. sample vertices $V^x_u = \{\hat{v}_1, \dots, \hat{v}_T\}$ and $V^z_u = \{\tilde{v}_1, \dots, \tilde{v}_T\}$ from the neighborhood $N(u)$ of $u$ in $V$.
        \item Randomly sample  $J \subseteq [N]$, and let $\ol{J} = [N]\setminus J$. 
        \item Define a point $\bx \in \mc{X}$ as follows. For each $i \in [M]$ and $v \in V$ set:
        \begin{equation}
            x_{v,i} = \begin{cases}
            1 & \tn{ if } v \in V^x_u \tn{ and } \pi_{vu}(i) \in J \\
            0 & \tn{ otherwise.}
            \end{cases}
        \end{equation}
        \item Define another point $\bz \in \mc{X}$ as follows. For each $i \in [M]$ and $v \in V$ set:
        \begin{equation}
            z_{v,i} = \begin{cases}
            1 & \tn{ if } v \in V^z_u \tn{ and } \pi_{vu}(i)\in \ol{J} \\
            0 & \tn{ otherwise.}
            \end{cases}
        \end{equation}
        \item Output $\left(B = \{\bx, \bz\}, \sigma_B = 1/2\right)$.
\end{enumerate}
\end{enumerate}
\medskip
{\bf YES Case.} This is easy to see using the same {\sf OR} formula $h^*$ defined in the previous section. Using the same arguments $h^*$ satisfies all the bags of size $2$ from $\mc{D}_{\mc{B}}$. Further, $h^*$ has exactly one (positive) literal from the coordinates corresponding to each $v \in V$ so that $h^*(\ol{\bx}^{(v)}) = 1$ where $\ol{\bx}^{(v)}$ is as defined in Step 2a. of $\mc{D}_{\mc{B}}$. Thus, $h^*$ satisfies all the bags of size $1$ as well.  

\subsection{NO Case}
Let us assume that there is an $\ell$-{\sf DNF} $h$ that satisfies $1/2 + \delta$ fraction of the bags of $\mc{D}_{\mc{B}}$. First, observe that if $h$ has a term consisting only of negated literals, then from the bi-regularity of $\mc{L}$ and a union bound that term will not have any coordinate from among the vertices chosen in Step 3a with probability at least $1 - 2T\ell/|V|$. Thus, $h$ will have label proportion $1$ on at least $1 - 2T\ell/|V|$ fraction of the $2$-sized bags i.e., $h$ will not satisfy them, implying that the maximum fraction of bags satisfied by $h$ is $1/2 + T\ell/|V|$. This is a contradiction since $|V| = \omega(T\ell)$ and can be taken to be large enough. Thus, we may assume that $h$ does not have a term of only negated literals.

Before proceeding, let us call a $v \in V$ as \emph{non-empty} if $h$ has a term in which all the positive literals correspond to $v$ and none of the negated literals correspond to $v$. Let $\Gamma(v)$ be an arbitrary map from each non-empty $v$ to one such term corresponding to it. Note that $\Gamma(v)$ is injective. Further, define $\Delta(v)$ to be the set of all indices $i \in [M]$ such that the positive literal corresponding to $(v,i)$ occurs in $\Gamma(v)$.

Define for each $u \in U$,  $\kappa_u$ to be the probability that the choice of $V^x_u$ and $V^z_u$ in Step 3a satisfies that there exists a pair $v, v' \in V^x_u \cup V^z_u$ of non-empty vertices such that $\pi_{vu}\left(\Delta(v)\right) \cap \pi_{v'u}\left(\Delta(v')\right) \neq \emptyset$ (call these \textit{intersecting non-empty} pair of vertices). Using this, let us define a randomized labeling $\rho$ for the vertices in $\mc{L}$ as follows: for each non-empty $v \in V$, $\rho(v)$ is chosen u.a.r from $\Delta(v)$, and for each $u \in U$, $v_u$ is chosen u.a.r. from $N(u)$ and then if $v_u$ is non-empty $\rho(u)$ set to $i$ which is chosen u.a.r. from $\Delta(v_u)$. By standard arguments, this labeling satisfies in expectation at least $\kappa_u/(4\ell^2T^2)$ fraction of the edges incident on $u$, and thus overall in expectation $\E_u[\kappa_u]/(4\ell^2T^2)$ fraction of edges of $\mc{L}$. Choosing the soundness of $\mc{L}$ to be small enough (and since $T$ is a constant), we can assume that $\E_u[\kappa_u] \leq \delta^2/200$.

By averaging, for at least $\delta/2$ fraction of $u \in U$ (call them \emph{good}) $h$ satisfies at least $1/2 + \delta/2$ fraction of the bags from $\mc{D}_{\mc{B}}\mid_u$ i.e., given the choice of $u$ in Step 1.  In particular from the above,
\begin{equation}
    \E_u[\kappa_u | u \tn{ good }] \leq \delta/100. \label{eqn:kappaugood}
\end{equation}

 For any such good $u$, $h$ must satisfy at least $\delta$ fraction of the bags of size $1$ and label proportion $1$ (as they constitute exactly half of the bags), implying that for at least $\delta$-fraction  of $v \in N(v)$, $h\left(\ol{\bx}^{(v)}\right) = 1$. Clearly, these vertices are non-empty for $h$ to evaluate to $1$ on them, and thus any good $u$ has at least $\delta$-fraction non-empty $v$ in $N(u)$.

Let us fix on one such good $u$. First, from above we can assume by adding an error probability of $\kappa_u$ that in Step 3a, $V^x_u \cup V^z_u$ does not contain an intersecting non-empty pair. Further, the probability that any term in $\{\Gamma(v)\,\mid\, v \in V^x_u \cup V^z_u\tn{ and non-empty}\}$ contains a negated literal corresponding to vertices in $V^x_u \cup V^z_u$ is at most $2T\ell/d_U$ where $d_U$ is the uniform degree on $U$. Since $d_U$ can be taken to be an arbitrarily large constant, possibly by replicating $V$, we can assume by adding an error probability of $\delta/100$ that this event does not occur.

Given the above, we will show that w.h.p over the choice of a $2$-sized bag, $h$ evaluates to $1$ on both $\bx$ and $\bz$.   We shall prove this for $\bx$, the argument for $\bz$ is analogous and the conjunction is obtained by a union bound. Since $u$ is good, the expected number of non-empty vertices in  $V^x_u$ is at least $\delta T$. By Chernoff bound, by adding an error probability of $\tn{exp}(-\delta T/8)$ we can assume there are at least $\delta T/2$ non-empty vertices in $V^x_u$. By our assumptions above, each term in $\{\Gamma(v)\,\mid\, v \in V^x_u \cup V^z_u\tn{ and non-empty}\}$ independently evaluates to $1$ w.p. at least $(1/2)^\ell$ over the choice of $\bx$. Thus, the probability that none of them evaluate to $1$ is at most $(1 - (1/2)^\ell)^{\delta T/2}$. This analysis can be repeated for $V^z_u$ and $\bz$.

Summing up the error probabilities (using repeated applications of \eqref{eqn:collect-error}), we obtain that:
\begin{equation}
    \Pr\left[h(\bx) \neq h(\bz)\right] \leq \kappa_u + \delta/100 + 2\left(\tn{exp}(-\delta T/8) + (1 - (1/2)^\ell)^{\delta T/2}\right) \leq \kappa_u + \delta/50 \label{eqn:kapp_uplusdeltaover50}
\end{equation}
for a good $u$, using an appropriate choice of $T = O\left(2^\ell\log(1/\delta)/\delta\right)$. By \eqref{eqn:kappaugood}, the average of the LHS of \eqref{eqn:kapp_uplusdeltaover50} over all good $u$ is at most $3\delta/100$, which means that $h$ satisfies on an average at most $1/2 + 3\delta/100$ bags corresponding to a random choice of good $u$. This is a contradiction to the definition of good $u$ thus completes the NO case analysis.

%% file: llp_parity.tex
\section{Hardness of LLP Learning Parities} \label{sec:paritieshardness}
This section is devoted to proving the following hardness reduction which, along with the inapproximability of \slc (Th. \ref{thm:slc-hardness}) proves Theorem \ref{thm:main-3}.
\begin{theorem}
    For any constants $\delta > 0$ and $q \in \mathbb{Z}^+$, $q \geq 2$, there is a polynomial time reduction from an instance $\mc{L}$ of \slc  to an \llpparity$[q]$ instance $\mc{B}$ s.t. \\
    \tn{YES Case:} If $\mc{L}$ is YES instance then there is a parity consistent with all the bags of $\mc{B}$. \\
    \tn{NO Case:} If $\mc{L}$ is a NO instance then there is no  parity satisfying at least $(q/2^{q-1} + \delta)$-fraction of the bags. 
\end{theorem}
We begin with the dictatorship test below using which the hardness reduction is described and analyzed in Sec. \ref{sec:hardness-parity}. 

\subsection{Dictatorship Test}
Consider a large $M \in \mathbb{Z}^+$, and the space of vectors $\F_2^M$. For some integer $q > 1$, let he dictatorship test distribution $\mc{D}^{\sf dict}_{M, q}$ on $(B, \sigma)$ be as follows:
\begin{enumerate}
    \item Choose $B = \left\{\bx^{(1)}, \dots, \bx^{(q)}\right\}$ where for each $i \in [M]$, $\left(x^{(1)}_i, \dots, x^{(q)}_i\right)$ is  sampled u.a.r. from $\left\{\mb{e}^{(j)}\right\}_{j=1}^q$ where $\mb{e}^{(j)} \in \F_2^q$ is the $j$th coordinate vector. 
    \item Output $(B, \sigma = 1/q)$.
\end{enumerate}
We prove the following lemma summarizing the completeness and soundness of t $\mc{D}^{\sf dict}_{M, q}$. 
\begin{lemma}\label{lem:parity_dict}
The distribution $\mc{D}^{\sf dict}_{M, q}$ satisfies the following properties:\\
\tn{Completeness:} For any $i \in [M]$ the parity function $h_i(\bx) := x_i$ has the property that $\sigma(B, h_i) = 1/q$ for any $B$ in the support of $\mc{D}^{\sf dict}_{M, q}$ i.e., $h_i$ satisfies all bags of $\mc{D}^{\sf dict}_{M, q}$. \\
\tn{Soundness:} Let $h(\bx) := c_0 \oplus \bigoplus_{i=1}^Mc_ix_i$ be such that $\left|\{i \in [M]\,\mid\, c_i = 1\}\right| = K \leq M$. Then, $h$ satisfies a random $(B, \sigma) \sim \mc{D}^{\sf dict}_{M, q}$ with probability at most $q/2^{q-1} + \exp\left(-2K/q + q/2\right)$. 
\end{lemma}
\begin{proof}
    The completeness follows from construction, since for any $(B, \sigma) \sim \mc{D}^{\sf dict}_{M, q}$, where $B = \left\{\bx^{(1)}, \dots, \bx^{(q)}\right\}$, $\left(h_i\left(x^{(j)}_i\right)\right)_{j=1}^q = \left(x^{(j)}_i\right)_{j=1}^q \in \left\{\mb{e}^{(j)}\right\}_{j=1}^q$ for all $i \in [M]$. Therefore, $\sigma(B, h_i) = 1/q = \sigma$, and therefore $h_i$ satisfies all the bags, for all $i \in [M]$.
    
    The proof of the soundness is given next in Sec. \ref{sec:proof_dict_soundness}.
\end{proof}
    
\subsubsection{Soundness of $\mc{D}^{\sf dict}_{M, q}$}\label{sec:proof_dict_soundness}
Let $I_h = \{i \in [M]\,\mid\, c_i = 1\}$ so that $|I_h| = K$. Letting $B = \left\{\bx^{(1)}, \dots, \bx^{(q)}\right\}$ be a random bag from $\mc{D}^{\sf dict}_{M, q}$, for convenience define the random variable $Z_j := h\left(\bx^{(j)}\right)$ for $j \in [q]$ and let $\mc{D}_Z$ be the distribution on $(Z_1, \dots, Z_q)$.
Since any $\left(B, \sigma\right)$ in the support of $\mc{D}^{\sf dict}_{M, q}$ has $\sigma = 1/q$, using the construction of $B = \left\{\bx^{(1)}, \dots, \bx^{(q)}\right\}$ we obtain
\begin{eqnarray}
    \bigoplus_{j=1}^q Z_j &  =  & \bigoplus_{j=1}^q h\left(\bx^{(j)}\right) = \bigoplus_{j=1}^q\left(c_0 \oplus \bigoplus_{i=1}^Mc_ix^{(j)}_i\right) \nonumber =  \bigoplus_{j=1}^q c_0 \oplus \bigoplus_{i\in I_h}\bigoplus_{j=1}^q x^{(j)}_i \nonumber \\ & = & \bigoplus_{j=1}^q c_0 \oplus \bigoplus_{i\in I_h} 1 = \bigoplus_{j=1}^q c_0 \oplus \bigoplus_{i=1}^K 1 = \psi^* \in \F_2  \label{eqn:parityofsumZj} 
\end{eqnarray}
Our goal is to show that $\mc{D}_Z$ is close to $\ol{\mc{D}}_q$ which we define to be the uniform distribution over the elements of $\F_2^q$ with parity $\psi^*$. Towards this we prove the following lemma, which shows that the distribution $\mc{D}_Z$ has low-bias.
\begin{lemma}\label{lem:halfbiasparity}
        Consider any strict non-empty subset $S \subsetneq [q]$, s.t. $1 \leq |S| = s \leq q$. Then, 
        \begin{equation*}
            \left|\Pr\left[\bigoplus_{j \in S} Z_j = 0\right] - \frac{1}{2}\right| \leq \frac{1}{2}\cdot\tn{exp}\left(-2K/q\right).
        \end{equation*}
\end{lemma}
\begin{proof}
    First, we may assume that $s \leq q/2$, otherwise we can use $[q]\setminus S$ along with \eqref{eqn:parityofsumZj} to complete the argument.
    Analogous to \eqref{eqn:parityofsumZj} we have that
    \begin{equation}
        \bigoplus_{j\in S} Z_j = \bigoplus_{j\in S} h\left(\bx^{(j)}\right) = \bigoplus_{j=1}^{|S|} c_0 \oplus \bigoplus_{i \in I_h}\bigoplus_{j \in S} x^{(j)}_i = \psi_{S,c_0} \oplus \bigoplus_{i \in I_h} r_i \label{eqn:sumofri}
    \end{equation}
    where $\psi_{S,c_0} = \bigoplus_{j=1}^{|S|} c_0$ is a constant and $r_i := \bigoplus_{j \in S} x^{(j)}_i$. From the construction of the random bag $B$, we have that $\{r_i\}_{i\in I_h}$ are iid $\F_2$-valued random variables such that $\Pr[r_i = 1] = s/q, \forall i \in I_h$. In other words, the RHS of \eqref{eqn:sumofri} denotes the parity of $K$ such iid random variables. To analyze this, let us consider an alternate way of sampling $\{r_i\}_{i=1}^M$: 
    \begin{enumerate}
        \item Sample $T \subseteq I_h$ by including each $i \in I_h$ into $T$ independently with probability $2s/q \leq 1$. 
        \item For each $i \in I_h\setminus T$, set $r_i = 0$. Independently for each $i \in T$, set $r_i = 1$ w.p. $1/2$ and to $0$ otherwise.
    \end{enumerate}
    It is easy to see that conditioned on $T \neq \emptyset$, $\oplus_{i \in I_h}r_i$ is unbiased. This, along with \eqref{eqn:sumofri} leads us to,
    \begin{equation}
        \Pr\left[\bigoplus_{j \in S} Z_j =  \psi_{S,c_0}\right] = \frac{1}{2}\cdot \left(1 - \Pr[T = \emptyset]\right) + p\cdot\Pr[T = \emptyset] = \frac{1}{2} + \Pr[T = \emptyset]\left(p - \frac{1}{2}\right) \label{eqn:pminushalf}
    \end{equation}
    where $p \in [0,1]$ is some probability. Further, 
    \begin{equation*}
        \Pr[T = \emptyset] = \left(1 - \frac{2s}{q}\right)^K 
        \leq \tn{exp}(-2Ks/q)
    \end{equation*}
     Since $s \geq 1$ and $|p-1/2| \leq 1/2$, the above along with \eqref{eqn:pminushalf} completes the proof.
\end{proof}

 The rest of the argument is similar to the Vazirani XOR Lemma except we need to show closeness to $\ol{\mc{D}}_q$ rather than the uniform distribution. We now transition to Fourier analysis of $\{0,1\}$-valued functions over $\{-1,1\}^q$. For this purpose, we shall map $\F_2$ to $\{-1,1\}$ via $b \mapsto (-1)^b$ and think of $\mc{D}_Z$ and $\ol{\mc{D}}_q$ as distributions over $\{-1,1\}^q$. First, from the definitions of $\mc{D}_Z$ and $\ol{\mc{D}}_q$, $\chi_{[q]}(\bz) = (-1)^{\psi^*}$ for any $\bz$ in the support of $\mc{D}_Z$ or $\ol{\mc{D}}_q$. Thus,
\begin{equation}
    \E_{\bz \leftarrow \ol{\mc{D}}_q}\left[\chi_{[q]}(\bz)\right] = \E_{\bz \leftarrow \mc{D}_Z}\left[\chi_{[q]}(\bz)\right] = \left(-1\right)^{\psi^*} \label{eqn:fullsetsame}
\end{equation}
and it is also easy to observe that for any $S \subsetneq [q]$,
\begin{equation}
    \E_{\bz \leftarrow \ol{\mc{D}}_q}\left[\chi_S(\bz)\right] = 0 \quad \tn{and} \quad  \left|\E_{\bz \leftarrow \mc{D}_Z}\left[\chi_S(\bz)\right]\right| \leq  \tn{exp}\left(-2K/q\right),
    \label{eqn:smallerror_uniformzero}
\end{equation}
where the upper bound follows from Lemma \ref{lem:halfbiasparity}.
Consider any function $f :\{-1,1\}^q \to [0,1]$ having Fourier expansion $\sum_{S \subseteq [q]}\hat{f}_S\chi_S$. Using \eqref{eqn:fullsetsame}, and \eqref{eqn:smallerror_uniformzero} we obtain
\begin{eqnarray*}
    \left|\E_{\ol{\mc{D}}_q}[f] - \E_{\mc{D}_Z}[f]\right| & \leq & \tn{exp}\left(-2K/q\right)\sum_{S\subsetneq [q]}\left|\hat{f}_S\right| \\ & \leq & 2^{q/2}\cdot \tn{exp}\left(-2K/q\right) \sqrt{\sum_{S\subsetneq [q]}\hat{f}_S^2} \,\leq\, 2^{q/2}\cdot \tn{exp}\left(-2K/q\right),
\end{eqnarray*}
where we use Cauchy-Schwarz and Parseval's bound. We can take $f$ to be indicator function of the event that exactly one of the coordinates is $-1$. This function evaluates to $1$ on $\ol{\mc{D}}_q$ with probability exactly $q/2^{q-1}$. Using this along with the above bound completes the proof.

\subsection{Hardness Reduction} \label{sec:hardness-parity}
Our hardness reduction is from an instance $\mc{L}$ of \slc given in Theorem \ref{thm:slc-hardness}.  

\medskip \noindent {\bf Points, bags and label proportions.} The initial set of points is defined in the space $\F_2^{V\times [M]}$. For a point $\hat{\bx} \in \F_2^{V\times [M]}$ let $\hat{\bx}[v] = (\hat{x}_{v,1}, \dots, \hat{x}_{v,M})$ be the vector of $M$ coordinates corresponding to $v \in V$. Let $\mc{D}_{\mc{B}}$ be the distribution on bags and label proportions given by the following process.
\begin{enumerate}
    \item Sample $v \in V$ u.a.r.
    \item Sample $(B = \{\bx^{(1)}, \dots, \bx^{(q)}\}, 1/q) \leftarrow \mc{D}^{\sf dict}_{M, q}$.
    \item For $j \in [q]$: define $\hat{\bx}^{(j)} \in \F_2^{V\times [M]}$ by letting $\hat{\bx}^{(j)}[v] = \bx^{(j)}[v]$ and for all $v' \neq v$, $\hat{\bx}^{(j)}[v'] = \mb{0}$.
    \item Output $(\hat{B} = \{\hat{\bx}^{(1)}, \dots, \hat{\bx}^{(q)}\}, 1/q)$
\end{enumerate}

\medskip
\noindent
{\bf Folding and projected point-set.} For each $e = (v_1, v_2) \in E$ and $j \in [N]$ define the linear constraint $C[e,j]$ over point $\hat{\bx} \in \F_2^{V\times [M]}$ as
\begin{equation}
    C[e,j] \Leftrightarrow \bigoplus_{i \in \pi_{ev_1}^{-1}(j)}\hat{x}_{v_1i} = \bigoplus_{i \in \pi_{ev_2}^{-1}(j)}\hat{x}_{v_2i}. \label{eqn:Cej}
\end{equation}
Let $H \subset \F_2^{V\times [M]}$ be the subspace of all the points which satisfy the set of homogeneous linear constraints $\mc{C} := \{C[e,j]\,\mid\, e \in E, j \in [N]\}$.
We let $H$ be the space in which our final instance resides by linearly projecting all points $\hat{\bx}$ created in the support of  $\mc{D}_{\mc{B}}$ into points $\ol{\bx} \in H$. Since our final instance is represented in a coordinate system corresponding to a linear basis for $H$, this also forces any solution $h$ to be represented in a basis for $H$. In particular, $h$ represented in the original space by $h(\hat{\bx}) := c_0 \oplus \langle \mb{c}, \hat{\bx}\rangle$ (where the inner product is over $\F_2$) must obey $\mb{c} \in H$.\vg{Should we justify the previous sentence?}\rs{Done.} Let $\ol{\mc{D}}_{\mc{B}}$ be the new distribution on the bags $(\ol{B}, 1/q)$ given by the linear projection of all the points in the bags of $\mc{D}_{\mc{B}}$ on to $H$. 

\subsubsection{YES Case}
In this case, there is a labeling $\rho : V \to [M]$ which satisfies all the edges of $\mc{L}$. Consider over $\F_2^{V\times [M]}$ the parity $h^*(\hat{\bx}) = \bigoplus_{v\in V}\hat{x}_{v,\rho(v)} =: \langle \mb{c}^*, \hat{\bx}\rangle$. Since, for any edge $e = (v_1, v_2)$, $\pi_{ev_1}(\rho(v_1)) = \pi_{ev_2}(\rho(v_2))$, $\mb{c}^* \in H$.
Now fix a choice of $v$ in Step 1 of the distribution $\mc{D}_{\mc{B}}$. Restricted to the coordinates corresponding to $v$ (since the others are set to $0$), $h^*$ is simply $x_{v,\rho(v)}$. We can now directly apply the completeness property of $\mc{D}^{\sf dict}_{M, q}$ in Lemma \ref{lem:parity_dict} to obtain that $h^*$ satisfies all the bags  given the choice $v$. Since this holds for all choices of $v$, $h^*$ satisfies all the bags of $\mc{D}_{\mc{B}}$.

\subsubsection{NO Case}
Assume for a contradiction that there is a parity in the space $H$ that satisfies $(q/2^{q-1} + \delta)$-fraction of the bags of $\mc{D}_{\mc{B}}$. This parity can be written as  
\begin{equation}
    h(\hat{\bx}) = c_0 \oplus \bigoplus_{v \in V}\bigoplus_{i=1}^Mc_i\hat{x}_i = c_0 \oplus \bigoplus_{v \in V}\langle \bc[v], \hat{\bx}[v]\rangle,
\end{equation}
where $\bc$ satisfies the constraints $\mc{C}$. By averaging there are $\delta/2$ fraction of \emph{good} $v \in V$ such that $h$ satisfies $(q/2^{q-1} + \delta/2)$-fraction of the bags of $\mc{D}_{\mc{B}}\mid_v$ i.e., $\mc{D}_{\mc{B}}$ given $v$ is chosen in Step 1. By the weak-expansion property of $\mc{L}$ in Theorem \ref{thm:slc-hardness}, the subset of edges $E'$ induced by the good vertices satisfies $|E'| \geq (\delta/2)^2|E|$. 
Let $S_v := \left\{i \in [M]\,\mid\, c_{v,i} = 1\right\}$.
From the soundness  of $\mc{D}^{\sf dict}_{M,q}$ (Lemma \ref{lem:parity_dict}), we obtain that all good $v$ satisfy
\begin{equation}
    \left|S_v\right| \leq \Delta :=  q(\log(2/\delta)+q/2)/2  \ . \label{eqn:Deltabd}
\end{equation}
The smoothness of $\mc{L}$ implies that for any good $v \in V$, $\Pr_{e\sim v}\left[\pi_{ev}\left(S_v\right) = \left|S_v\right|\right] \geq 1 - \left|S_v\right|^2/(2Q) \geq 1 - \Delta^2/(2Q)$. Let 
$E^* = \left\{e = (v_1, v_2) \in E'\,\mid\, \pi_{ev_r}\left(S_{v_r}\right) = \left|S_{v_r}\right|, r = 1,2\right\}$. 
Then
\begin{equation}
    \frac{\left|E^*\right|}{\left|E\right|} \geq \zeta := \frac{\delta^2}{4} - \frac{\Delta^2}{Q}. \label{eqn:estar}
\end{equation}
We have the following lemma.
\begin{lemma} \label{lem:intersectSv}
    For any $e = (v_1, v_2) \in E^*$,  $\pi_{ev_1}\left(S_{v_1}\right)\cap \pi_{ev_2}\left(S_{v_2}\right) \neq \emptyset$.
\end{lemma}
\begin{proof}
    Since $h$ satisfies at least one bag of $\mc{D}\mid_{v_1}$, $\bc[v] \neq \mb{0}$, and thus $S_{v_1}\neq \emptyset$. Consider any $j \in \pi_{ev_1}\left(S_{v_1}\right)$. From the definition of $E^*$, $\left|\pi_{ev_1}^{-1}(j) \cap S_{v_1}\right| = 1$. Thus, $\bigoplus_{i \in \pi_{ev_1}^{-1}(j)}\hat{x}_{v_1i} = 1$ and from \eqref{eqn:Cej} and the fact that $\hat{\bc}$ satisfies $ C[e,j]$ we obtain that $\pi_{ev_2}^{-1}(j) \cap S_{v_2} \neq \emptyset$.
\end{proof}
Let $\rho$ be the randomized labeling to the good vertices given by randomly assigning each good $v \in V$ a label chosen u.a.r. from $S_v$. From Lem. \ref{lem:intersectSv}, \eqref{eqn:Deltabd} and \eqref{eqn:estar} we obtain that $\rho$ satisfies in expectation at least, $\nu := \zeta/\Delta^2$ fraction of the edges of $\mc{L}$. By choosing the parameter $Q$ in Theorem \ref{thm:slc-hardness} to be large enough we can take $\zeta \geq \delta^2/8$ and then taking taking the parameter $R$ to be be large enough we obtain a contradiction. 

\section{Approximately LLP Learning Parities}\label{sec:algoparities} 
\noindent
We restate Theorem \ref{thm:main-algo} which is proved in the section.
\begin{theorem}
    For any $q \geq 2$, given $\mc{B} := \{(B_k, \sigma_k)\}_{k=1}^m$ as an instance of \llpparity$[q]$ over $\F_2^d$ such that there is an (unknown) parity  that satisfies all the bags of $\mc{B}$, there is a randomized polynomial time algorithm that satisfies $(1/2^{q-2})$-fraction of the bags of $\mc{B}$ in expectation. 
\end{theorem}

Let us first define the subsets of bags $\mc{B}_a := \{(B, \sigma) \in \mc{B}\,\mid\, \sigma = a\}$ for $a \in \{0,1\}$. We call the bags in $\mc{B}_0\cup \mc{B}_1$ as \emph{monochromatic} since we know that the vectors in any such bag are either all labeled $0$ or all labeled $1$. Therefore, one can write a (possibly non-homogeneous) $\F_2$-linear constraint (in the coefficients of the parity) for each vector in any monochromatic bag. Further, since the label proportion of each bag is given, the parity of labels in each bag is also determined. Thus, we can add these $\F_2$-linear constraints capturing the parity of the labels for each bag.

Since this system of linear equations is feasible (due to the existence of the satisfying parity) one can do Gaussian elimination to obtain a reduced instance $\mc{B}'$ of \llpparity$[q]$ which satisfies the following properties:
\begin{enumerate}
    \item $\mc{B}'$ has no monochromatic bags.
    \item A subset of the coefficients may be eliminated or  assigned a fixed value $\in \F_2$, and the rest are \emph{free}. 
    \item For any bag $(B, \sigma) \in \mc{B}'$, any assignment to the \emph{free} coefficients yields a labeling which satisfies the parity constraint of that bag. In particular, the set of such assignments yields a (possibly affine) subspace of labelings of size at most $2^{t-1}$ where $t = |B|$. Let us call this subspace of labelings as $F_B$. 
\end{enumerate}

The algorithm outputs a random parity given by a random assignment to the free coefficients. To analyze its performance, let us consider a bag $(B, \sigma) \in \mc{B}'$ where $|B| = t \leq q$ and $t\sigma \in \{1,\dots, t-1\}$ since $\mc{B}'$ has no monochromatic bags.
By feasibility there exists a vector in $F_B$ which has Hamming weight $t\sigma$. The probability that the bag will be satisfied by a random parity is precisely the probability that a random point in $F_B$ has Hamming weight $t\sigma$. There are two cases:

\begin{enumerate}
    \item \emph{$F_B$ contains all vectors of Hamming weight $t\sigma$}. Since $t\sigma \in \{1,\dots, t-1\}$ the number of such vectors is at least $t$. Since $|F| \leq 2^{t-1}$, the probability that the bag is satisfied is at least $t/2^{t-1} \geq  1/2^{q-2}$ for any positive integers $t \leq q$ and $q \geq 2$.
    \item \emph{$F_B$ does not contains all vectors of Hamming weight $t\sigma$}. In this case,  $F_B$ is at most $(t-2)$-dimensional and thus $|F_B| \leq 2^{t-2}$. Since $F_B$ does contain one vector of Hamming weight $t\sigma$, the probability that the bag is satisfied is at least $1/2^{t-2} \geq 1/2^{q-2}$.
\end{enumerate}